\documentclass[sigplan,10pt]{acmart}
\renewcommand\footnotetextcopyrightpermission[1]{} %

\setcopyright{none}

\usepackage{booktabs}
\usepackage{subcaption}
\usepackage[utf8]{inputenc}

\usepackage{amsmath} 
\usepackage[capitalise,noabbrev]{cleveref}
\usepackage{amssymb}
\usepackage{amsthm,amstext,amsfonts,bbm}
\usepackage{amstext,amsfonts,bbm}
\usepackage{algorithm,algorithmicx,graphicx,xspace,nicefrac} 
\usepackage{mathtools}
\usepackage{comment} 
\usepackage{epsfig} 

\usepackage{enumitem}
\newcommand{\defeq}{\stackrel{\textup{def}}{=}}

\newtheorem{problem}{Problem}
\newtheorem{remark}{Remark}
\newtheorem{assumption}{Assumption}

\newcommand{\tv}{\mathrm{tv}}

\def\abs#1{\left| #1 \right|}
\newcommand{\norm}[1]{\ensuremath{\left\lVert #1 \right\rVert}}

\newcommand{\marginlabel}[1]%
{\mbox{}\marginpar{\it{\raggedleft\hspace{0pt}#1}}}

\newcommand{\veps}{\varepsilon}

\newcommand\calF{\mathcal{F}}

\newcommand\calL{\mathcal{L}}
\newcommand\calM{\mathcal{M}}

\newcommand\bbF{\mathbb{F}}

\def\implies{\Rightarrow}

 {
	\begin{enumerate}}{\end{enumerate}}

\def\qedsketch{\ifmmode\Box\else{\unskip\nobreak\hfil
\penalty50\hskip1em\null\nobreak\hfil$\Box$
\parfillskip=0pt\finalhyphendemerits=0\endgraf}\fi}

\usepackage[stable]{footmisc}

\newcommand\MH{\mathrm{MH}}

\newcommand\langN{\calL_\Normalize}
\newcommand\langS{\calL_\Stationary}

\newcommand{\restrict}[2]{\left. #1 \right\rvert_{#2}}

\usepackage[utf8]{inputenc}
\usepackage[T1]{fontenc}
\usepackage{amsmath, amsfonts}
\usepackage{mathtools}
\usepackage{semantic}
\usepackage{turnstile}
\usepackage{stmaryrd}
\reservestyle[\mathinner]{\command}{\mathbf}
\command{let[let\;], in[\;in\;]}
\command{case[case\;], of[\;of\;]}
\command{if[if\;] , then[\;then\;]}
\command{match[match\;] , with[\;with\;]}
\command{sample, score, return, stat, Iterate, flip, norm, uniform, bernoulli, coin,  Prior, Lhd, Tracer}
\newcommand{\Let}{{\<let>}}
\newcommand{\In}{{\<in>}}
\newcommand{\Case}{{\<case>}}
\newcommand{\Of}{{\<of>}}

\newcommand{\Match}{{\<match>}}
\newcommand{\With}{{\<with>}}
\newcommand{\Return}{{\<return>}}
\newcommand{\Sample}{{\<sample>}}

\newcommand{\Score}{{\<score>}}
\newcommand{\Normalize}{{\<norm>}}
\newcommand{\Stationary}{{\<stat>}}
\newcommand{\Iterate}{{\<Iterate>}}
\newcommand{\Prior}{{\<Prior>}}
\newcommand{\Tracer}{{\<Tracer>}}
\newcommand{\Likelihood}{{\<Lhd>}}

\newcommand{\TypeA}{\mathbb{A}}
\newcommand{\TypeB}{\mathbb{B}}

\global\long\def\Reals{\mathbb{R}}
\global\long\def\Nats{\mathbb{N}}

\global\long\def\NNReals{\Reals_{+}}
\global\long\def\Unit{\mathbf{1}}

\global\long\def\BernoulliDist{\mathrm{Bern}}
\ifdefined\set
\renewcommand{\set}[1]{\ensuremath{\left\lbrace #1 \right\rbrace}}
\else
\newcommand{\set}[1]{\ensuremath{\left\lbrace #1 \right\rbrace}}
\fi

\mathlig{|-p}{\sststile[s]{p}{}}
\mathlig{|-d}{\sststile[s]{d}{}}
\mathlig{|-c}{\sststile[s]{c}{}}
\mathlig{|-z}{\sststile[s]{z}{}}
\mathlig{|-z1}{\sststile[s]{z_1}{}}
\mathlig{|-z2}{\sststile[s]{z_2}{}}
\mathlig{|-z3}{\sststile[s]{z_3}{}}
\mathlig{|-1}{\sststile[s]{p1}{}}
\mathlig{|-}{\sststile[s]{}{}}
\renewcommand{\compiler}{\Phi}

\newcommand{\denote}[1]{\left\llbracket #1 \right\rrbracket}

\newenvironment{denotes}%
{\left\llbracket}{\right\rrbracket}

\mathlig{[[}{\begin{denotes}}
\mathlig{]]}{\end{denotes}}

\newcommand{\defn}[1]{\emph{#1}}
\newcommand{\probabilistic}{probabilistic}
\newcommand{\uprobabilistic}{purely probabilistic}

\begin{document}

\title{Approximations in Probabilistic Programs}
\author{Ekansh Sharma}
\email{ekansh@cs.toronto.edu}
\affiliation{
  \department{Department of Computer Science}
  \institution{University of Toronto}
}
\affiliation{
  \institution{Vector Institute}
}

\author{Daniel M. Roy}
\email{droy@utstat.toronto.edu}
\affiliation{
  \department{Department of Statistical Sciences}
  \institution{University of Toronto}
}
\affiliation{
  \institution{Vector Institute}
}

\begin{abstract}
We study the first-order probabilistic programming language introduced by \citet{SWY+:16}, but with an additional language construct, $\Stationary$, that, 
like the fixpoint operator of \citet{AYC:18}, converts the description of the Markov kernel of an ergodic Markov chain 
into a sample from its unique stationary distribution.
Up to minor changes in how certain error conditions are handled,
we show that $\Normalize$ and $\Score$ are eliminable from the extended language, in the sense of \citet{Fel:91}.
We do so by giving an explicit program transformation and proof of correctness. 
In fact, our program transformation implements a Markov chain Monte Carlo algorithm, 
in the spirit of the ``Trace-MH'' algorithm of \citet{WSG:11} and \citet{GMR+:08}, but less sophisticated to enable analysis. 
We then explore the problem of approximately implementing the semantics of the language with potentially nested $\Stationary$ expressions,
in a language without $\Stationary$.
For a single $\Stationary$ term, the error introduced by the finite unrolling 
proposed by Atkinson et al.\ vanishes only asymptotically.
In the general case, no guarantees exist.
Under uniform ergodicity assumptions, we are able to give quantitative error bounds and convergence results for the approximate implementation of the extended first-order language.
\end{abstract}


\maketitle

\section{Introduction}
\label{sec:intro}
Probabilistic programming languages (PPLs) for Bayesian modelling, like Stan and TensorFlow Probability, provide statisticians and data scientists a formal language to model observed data and latent (i.e., unobserved) variables \citep{tensorflow2015-whitepaper, JSSv076i01,GMR+:08}. In these languages, users specify a ``prior'' probability distribution that represents prior beliefs/assumptions about the data and latent variables. Then conditional statements are introduced to represent actual observed data. The semantics of a program in these probabilistic languages is the ``posterior'' (i.e., conditional) probability distribution, which represents the prior beliefs after being updated (conditioned) on observed data.

\citet{SWY+:16,Sta:17} give precise measure theoretic semantics to an idealized first order probabilistic languages. Beyond the standard deterministic core language, the language provides three constructs for modeling data and latent variables: 
the $\Sample$ construct introduces latent variables whose values are subject to uncertainty;
the $\Score$ construct introduces data into a model in the form of likelihoods;
the $\Normalize$ construct combines the latent variables and data likelihoods to produce the (normalized) posterior distribution, representing updated beliefs given data.

Irrespective of the theoretical prospect of exact implementations of PPL semantics, approximations are ubiquitous in implementations of PPLs that are intended to scale to the demands of real-world Bayesian modelling.
Indeed, in all but the simplest models, approximations necessarily arise because the key operation---computing the posterior distribution from the prior distribution and data---is computationally intractable.
As a result, every scalable PPL implementation 
performs a nontrivial transformation that modifies the semantics of the original program in complex ways by introducing errors that tradeoff with computation.
Often, these transformations are implicit within an inference engine, rather than producing a new program in a formal language.
Since approximation is critical to scalable implementations of PPLs, 
it is of interest to define PPLs in which we can reason about the process of introducing approximations. 
In this work, we take a step in this direction.

Most practical implementations of PPLs use one (or a combination) of two classes of approximate inference methods:
\emph{Markov chain Monte Carlo} (MCMC) methods approximate the posterior via the simulation of a Markov chain whose stationary (target) distribution is the posterior distribution;
whereas, \emph{variational} methods use (possibly stochastic) optimization to search among a family of tractable distribution for one that best approximates the posterior distribution.

In this work, we focus on MCMC methods.
While many existing PPLs are powerful enough to represent the MCMC implementation of any program in their own language, existing systems cannot reason about the error introduced by simulating Markov chains for a finite number of steps.
This challenge grows when we ask how the error scales under composition of multiple approximate programs and nested MCMC approximations.

\subsection{Contributions}

In this paper, we bridge the gap between the semantics of a probabilistic programming language
and that of a Markov chain Monte Carlo implementation of the language,
building on the foundations laid by \citet{SWY+:16,Sta:17}.
We start from their first-order PPL,
whose $\Sample$, $\Score$, and $\Normalize$ constructs permit the specification of Bayesian models and conditioning.
In \cref{sec:semantics}, we append the language proposed by \citet{SWY+:16} to include $\Stationary$ construct, that takes a Markov kernel and returns the unique stationary distribution, if one exists. 
In \cref{sec:equivalence}, using the framework of eliminability \citep{Fel:91},
we show that the language with only $\Sample$ and $\Stationary$ is equivalent to the original language.
In \cref{sec:stationary}, we give an approximate compiler for the $\Stationary$ construct, similar to \citet{AYC:18}. We then identify some semantic constraints on the Markov kernel, given as argument to the $\Stationary$ construct, 
under which we can derive quantitative error bounds.

\section{Related Work}
\label{sec:related}

\citet{AYC:18} propose a language in which users can manually construct inference procedures for probabilistic models. In particular, users can specify (Markov transition) kernels and their stationary distributions using a fixpoint construct, like $\Stationary$. They propose to approximately compute fixpoints by iteration, but leave open the problem of characterizing how this approximation affects the semantics, and whether one can control the approximation error. Our work provides the first such guarantees by exploiting uniform ergodicity.

\citet{HNR+:15,BDGS:16,SKV+:18} introduce Markov Chain Monte Carlo based implementations of distinct but related PPLs. In each case, the authors provide theoretical guarantees, building on standard results in Markov chain theory.
Informally speaking, when a program has the equivalent of a single, top-level $\Normalize$ expression, these results guarantee that the original semantics are recovered asymptotically.
This line of work does not provide a framework for quantifying error. As a result, these results do not bear on the correctness, even asymptotically, of programs with nested $\Normalize$ terms. In \cref{sec:stationary}, we give  assumptions on the underlying Markov chains that allow us to quantify the error  of the approximate inference scheme, even under composition and nesting.

\citet{Rain:18} studies the error associated with Monte Carlo approximations to nested $\Normalize$ terms, assuming that one can produce exact samples from the $\Normalize$ terms. In this work, 
we contend with the approximation error associated with nested Markov chains that are not assumed to have converged to their stationary distribution. 
In \cref{sec:stationary}, 
we show that one can quantify the rate at which nested Markov chains converge, under additional hypotheses. 

In the MCMC literature, \citet{RRS:98,MRS:19} study the convergence of Markov chain with approximate transition kernels. \citet{MRS:19} gives quantitative convergence bounds for approximate Metropolis--Hastings algorithm when the acceptance probability is approximate, e.g., due to the dynamics of the Markov chain being perturbed. In our setting, the perturbation of the Markov chain is caused by the dynamics of nested Markov chains. We posit assumptions of uniform ergodicity on nested Markov chains and give quantitative error bounds. 

\section{Measure-theoretic preliminaries}
\label{sec:preliminaries}

We assume the reader is familiar with the basics of measure and probability theory: $\sigma$-algebras, measurable spaces, measurable functions, random variables, (countably additive) measures, (Lebesgue) integration, conditional probability distributions. 
We review a few definitions here for completeness:
Recall that a measure $\mu$ on a measurable space $(X,\Sigma_X)$ is said to be \defn{finite} when $\mu(X)<\infty$, said to be \defn{S-finite} when $\mu = \sum_{i\in \Nats}\mu_i$, 
where $\mu_i$ is a finite measure for all $i$, and said to be \defn{$\sigma$-finite} when $\mu = \sum_{i\in \Nats}\mu_i$, where $\mu_i$ is a finite measure for all $i$ and $\mu_i,\mu_j$ are \defn{mutually singular} for all $i\neq j$. We say $\mu$ is a \defn{probability measure} when $\mu(X)= 1$ and a \defn{sub-probability measure} when $\mu(X)\leq 1$. We say a property $P$ holds \defn{$\mu$-almost everywhere (or $\mu$-a.e.)} if $\mu(\set{x\colon \neg P(x)}) = 0$. 

Let $(X, \Sigma_X)$ and $(Y, \Sigma_Y)$ be measurable spaces.
Recall that a \defn{kernel (from $X$ to $Y$)} is a function $k\colon X\times\Sigma_{Y} \to \NNReals$ 
such that, for  all $x\in X$, the function $k(x, \cdot)\colon \Sigma_{X} \to [0,\infty]$ is a measure, and, for all $A\in \Sigma_X$, the function $k(\cdot, A)\colon  \to [0,\infty]$ is measurable. 
We say that $k$ is \defn{finite} (resp., \defn{{$\sigma$-finite}}, and \defn{{S-finite}}) if, for all $x$, $k(x, \cdot)$ is a finite (resp., \defn{{$\sigma$-finite}}, and \defn{{S-finite}}) measure.
Similarly, $k$ is a \defn{(sub-)probability kernel} if for all $x$, $k(x, \cdot)$ is a (sub-)probability measure.
In this paper we study approximations to probabilistic semantics and give rates of convergence. 
These notions require metric structure on the space of probabilistic distributions. 
Recall that the \defn{total variation} distance between a pair $\mu,\nu$ of probability distributions on a common measurable space $(X,\Sigma_X)$ is defined to be 
  \[
  \norm{\mu-\nu}_{\tv} = \sup_{A\in \Sigma_{\denote{\TypeA}}}  \abs{\mu(A) - \nu(A)}.
  \]
The total variation distance is a metric.

\subsection{Markov Chain theory}
\label{sec:MCT}
We introduce a new language construct explicitly designed to refer to the long run behavior of repeating the same probabilistic computation over and over. The study of such phenomena falls within the remit of Markov chain theory. We summarize basic definitions and results here.

Fix a basic probability space, and write $P$ for the associated probability measure. All random variables will be defined relative to this structure.
Let $(X, \Sigma_X)$ be a measurable space. A \defn{Markov chain} (with \defn{state space} $X$) is a sequence of $X$-valued random variables $(Y_0, Y_1 , \dots)$  such that, for some probability kernel $k:X\times \Sigma_X \to [0,1]$, with probability one, for all measurable subsets $A \in \Sigma_X$,
  \[
  P(Y_i \in A| Y_0,Y_1,\dots,Y_{i-1}) =  k(Y_{i-1}, A).
  \]  
  The probability kernel $k$ is called the \defn{transition kernel} of the Markov chain. 
  The transition kernel $k$, combined with an initial distribution $\mu = P(Y_0 \in \cdot)$, determines the distribution of the Markov chain.
  Note that our definition is somewhat stronger than the usual abstract definition in terms of conditional independence, though agrees with the usual definition for Borel spaces.

  For notational convenience, define $k^n:X\times\Sigma_X \to [0,1]$ by
  \[
  k^n(x, \cdot) \defeq \int_X k(y, \cdot)k^{n-1}(x, dy).
  \]
  The (probability) distribution (or law) of the random variable $Y_n$, 
  denoted $\calL(Y_n)$, satisfies
  \[
  \calL(Y_n)(\cdot) = \int k^n(x,\cdot) \mu(dx).
  \]
We say that a probability measure $\pi$ is \defn{invariant} with respect to $k$ when
  \[
  \norm{\pi(\cdot) - \int_X k(x, \cdot) \pi(dx)}_\tv = 0.
  \]
A Markov chain is \defn{ergodic} if it admits a unique invariant distribution, $\pi$. 
Recall that $\mu$ is \defn{absolutely continuous} with respect to $\pi$  (equivalently, $\pi$ dominates $\mu$, or $\pi \gg \mu$), when, for all $E \in \Sigma_X$, $\pi(E) = 0$ implies $\mu(E)=0$.
When $\mu$ is absolutely continuous with respect to $\pi$, then
  \begin{equation}
  \label{eq:ergodic}
\norm{\int_X k^n(x, \cdot) \mu(dx) - \pi(\cdot)}_\tv \to 0 \qquad \text{as $n \to \infty$.}      
  \end{equation}

In this case, $\pi$ is called the \defn{stationary distribution}.
The Markov chain is \defn{uniformly ergodic} if there exist constants $C\in \NNReals, \rho\in (0,1)$ such that
\begin{equation}
\label{eq:unifergodic}
 \sup_{x \in X} \, \norm{k^n(x, \cdot) - \pi(\cdot)}_\tv \leq C\rho^n.    
\end{equation}

\subsection{The Metropolis--Hastings algorithm}

The Metropolis--Hastings (MH) algorithm is a framework for constructing Markov chains with a specified (``target'') stationary distribution. 
In particular, given a probability kernel $Q$ from a measurable space $(X, \Sigma_X)$ to itself, called the \defn{proposal kernel},
the MH algorithm describes how to modify the proposal kernel in order to get a new kernel that defines a Markov chain with the desired stationary distribution. 
In particular, the \defn{Metropolis--Hastings kernel} proposes to transition from state $x \in X$ 
by sampling a potential state transition $y$ from $Q(x,\cdot)$.
The kernel accepts the transition with probability $\alpha(x,y)$, where $\alpha\colon X\times X\to [0,1]$ is a carefully crafted measurable function. If the potential transition is rejected, the algorithm instead chooses the self-transition from $x$ to itself.
The Metropolis--Hastings kernel leaves $\pi$ invariant if and only if, for all bounded measurable functions $f$, 
\begin{equation}
  \label{eq:detailedbalance}
  \begin{aligned}
  &\iint f(x,y) \pi(dx)Q(x, dy)  \\
  & =\iint f(x,y) \pi(dy)Q(y, dx)\alpha(y,x).
  \end{aligned}
\end{equation}
The above equation is known as \defn{detailed balance}. 
It implies that the Markov chain is time reversible.
 
\begin{theorem}[Radon--Nikodym for kernels] \label{rnthm}
Let $(X, \Sigma_X)$ be a measurable space where $\Sigma_X$ is Borel, $\mu\colon X\times \Sigma_X\to [0,\infty]$ be an $\sigma$-finite kernel, and $\nu\colon X\times\Sigma_X$ be a $\sigma$-finite kernel. If $\mu$ is absolutely continuous with respect to $\nu$,  then there exists a measurable function $f\colon X\times X\to [0,\infty)$ such that, for all $x\in X$ and $A\in \Sigma_{X}$,
\[
  \mu(x, A) = \int_{A}f(x, y)\nu(dy).
\] 
Furthermore, $f$ uniquely defined up to a $\nu$-null set, i.e., if $\mu(x,A) = \int_A g \,d\nu$ for all $A \in \Sigma_{X}$, then $f=g$ $\nu$-almost everywhere.
\end{theorem}
By the almost everywhere uniqueness, we are justified in referring to $f$ as \emph{the} Radon--Nikodym derivative of $\mu$ with respect to $\nu$.

\begin{theorem}[\citet{Tie:98}] 
  \label{thm:tie98}
  For the Metropolis--Hastings algorithm with proposal kernel $Q$ and target stationary distribution $\pi$,
  assume there exists a measure $\nu$ that dominates  both $\pi$ and $Q(x,\cdot)$, for all $x \in X$.
  Writing $p$ and $q(x,\cdot)$ for the Radon--Nikodym derivatives of $\pi$ and $Q(x,\cdot)$ with respect to $\nu$. 
 Let $r(x,y) = \frac{p(x)q(x,y)}{p(y)q(y,x)}$, let $$R = \set{(x,y) : p(x)q(x,y) > 0\text{ and }p(y)q(y,x)>0},$$  and define
  \[
    \alpha(x, y) = \begin{cases}
      \min\set{1,r(y,x)}, \quad \text{if } (x, y)\in R,\\
      0, \quad \text{otherwise}.
    \end{cases}
  \]
  Then the Metropolis--Hastings algorithm satisfies \cref{eq:detailedbalance}.
\end{theorem}

\section{Language Syntax and Semantics}
\label{sec:semantics}
In this section, we present an idealized first-order probabilistic language with a construct $\Stationary$ that takes as input a transition kernel for a Markov chain on some state space and returns the stationary distribution associated with the Markov chain. 
The language we present is based on the first-order probabilistic language introduced and studied by \citet{SWY+:16} and \citet{Sta:17}, 
which has constructs for sampling, soft constraints, and normalization. 
The key differences, which we highlight again below, are (i) a syntactic distinction between probabilistic terms with and without soft constraints, which affects also typing, (ii) error cases in the denotation of $\Normalize$, and (iii) the introduction of the new construct, $\Stationary$.

\subsection{Probabilistic language with construct for stationary}
We begin with types and syntax of the language, presented in \cref{fig:lang1}. For the remainder of this paper, we call this language $\langN$.

\begin{figure}[]
  \begin{align*}
    &\defn{Types: } \\
    &\quad\TypeA_0, \TypeA_1 ::= \Reals \mid \Unit  \mid  \mathrm{P}(\TypeA) \mid  \TypeA_0 \times \TypeA_1 \mid \sum_{i\in \Nats}\TypeA_i\\
    &\defn{Terms: } \\
    &\quad\defn{deterministic:} \\
    &\quad a_0,a_1  ::= x \mid * \mid (a_0,a_1) \mid(i,a) \mid \pi_j(a) \mid f(a)\\
    &\quad \qquad \quad \mid \Case a \Of \set{(i,x)\Rightarrow a_i}_{i\in I} \\
    &\quad\textit{\uprobabilistic{}:}\\
    &\quad t_0,t_1  ::= \Sample (a) \mid \Return (a) \mid \Let x=t_0 \In t_1 \\ 
    &\quad\qquad \quad         \mid \Case a \Of \set{(i,x)\Rightarrow t_i}_{i\in I} \\
    &\quad\qquad \quad         \mid \Stationary (t_0, \lambda x.t_1) \mid \Normalize (v)\\
    &\quad\textit{probabilistic:}\\
    &\quad v_0,v_1  ::= t \mid \Let x=v_0 \In v_1 \\ 
    &\quad\qquad \quad         \mid \Case a \Of \set{(i,x)\Rightarrow v_i}_{i\in I} \\
    &\quad\qquad \quad         \mid \Score(a) \\
    &\defn{Program: }\\
    &\quad \text{$t$ is a program if $t$ is \uprobabilistic{}}\\
    & \qquad\qquad \text{ with no free variables}
    \end{align*}
  \caption[]{Syntax for the probabilistic language: $\langN$}
  \label{fig:lang1}
\end{figure}

\subsection*{Types}
We study a typed probabilistic programming language where every term in the language has a type generated by the following grammar:
\begin{align*}
\TypeA_0, \TypeA_1 ::= \Reals \mid \Unit  \mid  \mathrm{P}(\TypeA) \mid  \TypeA_0 \times \TypeA_1 \mid \sum_{i\in I}\TypeA_i,
\end{align*}
where $I$ is some countable, non-empty set. 
The language has a standard unit type, product types, and sum types. 
Denotationally, types are interpreted as measurable spaces, i.e., each type $\TypeA$ represents some set $\denote{\TypeA}$ coupled with a $\sigma$-algebra  $\Sigma_{\denote{\TypeA}}$ associated with the set $\denote{\TypeA}$. 
We'll use $\denote{\TypeA}$ as shorthand for the measurable space $(\denote{\TypeA}, \Sigma_{\denote{\TypeA}})$.
The language has a special type $\Reals$ for real numbers and, for each type $\TypeA$, a type $P(\TypeA)$ for probability distributions on (the measurable space denoted by) $\TypeA$.

Now we give the space and the $\sigma$-algebra associated to each type and type constructor:
\begin{itemize}
  \item For the type $\Reals$, the underlying space $\denote{\Reals}$ is the space of real numbers; $\Sigma_{\denote{\Reals}}$ is the (standard Borel) $\sigma$-algebra generated by the set of open subsets on the real line.
  \item For the unit type $\Unit$, the underlying space is $\set{()}$; $\Sigma_{\denote{\Unit}} = \set{\set{()}, \varnothing}$ is the corresponding $\sigma$-algebra.
  \item For the type $\mathrm{P}(\TypeA)$, the underlying space is the set of probability measure on $\denote{A}$, denoted $\calM(\denote{\TypeA})$; 
  $\Sigma_{\calM(\denote{\TypeA})}$ is the the $\sigma$-algebra generated by sets $\set{\mu~|~\mu(A)\leq r}$ for all $A \in \Sigma_{\denote{\TypeA}}$ and $r\in [0,1]$.
  \item For the product type $\TypeA \times \TypeB$, the underlying space is  the product space $\denote{\TypeA\times \TypeB} = \denote{\TypeA}\times \denote{\TypeB}$; $\Sigma_{\denote{\TypeA\times \TypeB}}$ is the (product) $\sigma$-algebra generated by the rectangles $U\times V$, for $U \in \Sigma_{\denote{\TypeA}}$ and $V\in \Sigma_{\denote{\TypeB}}$. 
  \item For the sum type $\sum_{i\in I}\TypeA_i$, the underlying space is the co-product space 
  $\denote{\sum_{i\in I} \TypeA_i} = \biguplus_{i\in I} \denote{\TypeA_i}$; 
  informally, $\Sigma_{\denote{\sum_{i\in I} \TypeA_i}}$ is the $\sigma$-algebra generated by the sets of the form $\set{(i,a) ~|~ a\in U}$ for $U\in \Sigma_{\denote{\TypeA_i}}$.
\end{itemize}

\subsection*{Terms}
As in \citet{SWY+:16},
each term in the language is either \defn{deterministic} or \defn{\probabilistic{}}, 
satisfying typing judgments of the form 
$\Gamma |-d t\colon\TypeA$ and 
$\Gamma |-p t\colon\TypeA$, respectively, given some environment/context $\Gamma = (x_1\colon\TypeA_1, ..., x_n\colon\TypeA_n)$.
Letting $\denote{\Gamma} = \prod_{i=1}^n \denote{\TypeA_i}$, 
a deterministic term denotes a measurable function from the environment $\denote{\Gamma}$ to $\denote{\TypeA}$.
As in \citet{Sta:17}, a probabilistic term denotes an $S$-finite kernel from $\denote{\Gamma}$ to $\denote{\TypeA}$. 
Different from \citet{SWY+:16,Sta:17}, we distinguish a subset of probabilistic terms we call \defn{\uprobabilistic{}},
which satisfy an additional typing judgment $\Gamma |-1 t\colon\TypeA$. 
A \uprobabilistic{} term denotes a probability kernel from $\denote{\Gamma}$ to $\denote{\TypeA}$.
Departing again from \citet{SWY+:16, Sta:17},
a \defn{program} in our language is a \uprobabilistic{} term with no free variables.

\subsubsection{Variables, measurable functions, constructors, and destructors}

As in \citet{SWY+:16}, 
the language contains standard variables, constructors, and destructors, and constant terms $f$ for all the measurable functions $f\colon\denote{\TypeA_0}\to \denote{\TypeA_1}$. 
The typing rules and semantics are unchanged and reproduced here for completeness:
\begin{align*}
  &\quad\inference*{}{\Gamma, x\colon\TypeA, \Gamma' |-d x\colon \TypeA} \quad\quad   
  \inference*{}{\Gamma |-d ()\colon \Unit} 
  \\
  &\quad \quad\;\;  \inference*{\Gamma |-d t\colon \TypeA_i}{\Gamma |-d (i, t)\colon \sum_{j\in I}\TypeA_j}\ (i \in I)\\
 &\inference*{\Gamma |-d t\colon \sum_{i\in I}\TypeA_i \quad (\Gamma, x\colon \TypeA_i |-z u_i\colon \TypeB)_{i\in I}}{\Gamma|-z \Case t \Of \set{(i,x) \implies u_i}_{i\in I}\colon \TypeB}(z\in \set{d,p, p1})\\
  & \inference*{\Gamma |-d t_0\colon \TypeA_0 \quad \Gamma |-d t_1\colon \TypeA_1}{\Gamma |-d (t_0,t_1)\colon \TypeA_0 \times \TypeA_1} \quad \inference*{\Gamma |-d t\colon \TypeA_0}{\Gamma |-d f(t)\colon \TypeA_1} \\
  & \inference*{\Gamma |-d t \colon \TypeA_0\times \TypeA_1}{\Gamma |-d \pi_j(t)\colon \TypeA_j} (j\in \set{0,1}).
\end{align*}
\begin{align*}
&\denote{x}_{\gamma, d, \gamma'} \defeq d,\quad \denote{()}_{\gamma} \defeq (), \quad \denote{(i,t)}_\gamma \defeq (i, \denote{t}_\gamma), \\ 
&\denote{(t_0,t_1)}_\gamma \defeq (\denote{t_0}_\gamma, \denote{t_1}_\gamma), \quad \denote{f(t)}_\gamma \defeq f(\denote{t}_\gamma), \\
&\denote{\pi_j(t)}_\gamma \defeq d_j \;\; \text{if } \denote{t}_\gamma \defeq(d_0,d_1).
\end{align*}
For case statements, if the kind judgement gives a deterministic term, the semantics is
\[\denote{\Case t \Of \set{(i,x) \implies u_i}_{i\in I}}_\gamma \defeq \denote{u_i}_{\gamma, d} \;\; \text{if } \denote{t}_\gamma = (i,d).
\]
If the kind judgement gives a probabilistic term, the term denotes a $S$-finite kernel and the semantics is
\begin{align*}
  &\denote{\Case t \Of \set{(i,x) \implies u_i}_{i\in I}}_{\gamma, A}  \\
  &\qquad\qquad \defeq \denote{u_i}_{\gamma, d, A} \;\; \text{if } \denote{t}_\gamma = (i,d),    
\end{align*}
where $A \in \Sigma_{\denote{\TypeB}}$.

\subsubsection{Sequencing and sampling terms} 
In addition to standard $\Let$statements and $\Return$ statements for sequencing, 
the language has a construct for producing a random sample from a probability distribution. 
The typing rules are as follows:
\begin{align*}
  \inference*{\Gamma |-z1 t_1\colon\TypeA \quad \Gamma,x:\TypeA |-z2 t_2\colon\TypeB}{\Gamma |-z3 \Let x=t_1 \In t_2 \colon\TypeB},
\end{align*}
where $z_3 = \begin{cases}
    p1 \quad \text{if }z_1 = p1\wedge z_2 =p1\\
    p \quad\; \text{if }z_1 = p\vee z_2 =p
  \end{cases}.$
 \begin{align*}
  \inference*{\Gamma |-d t\colon\TypeA }{\Gamma |-1 \Return(t)\colon\TypeA} \qquad \inference*{\Gamma |-d t\colon\mathrm{P}(\TypeA) }{\Gamma |-1 \Sample(t)\colon\TypeA}
\end{align*}

Now all the terms for sequencing are judged as probabilistic terms, where $\Sample$ and $\Return$ are in fact probabilistic  measure one terms. Also, if both $t_1$ and $t_2$ of the $\Let$statements are judged as \uprobabilistic{} terms then so is the $\Let$statement. 

As in \citet{Sta:17}, the semantics of the $\Let$construct is defined in terms of integration as follows:
\[
  \denote{\Let x=t_1\In t_2}_{\gamma, A} \defeq \int_{\denote{\TypeA}} \denote{t_2}_{\gamma,x, A} \denote{t_1}_{\gamma, dx}
\]
Since both $t_1$ and $t_2$ are probabilistic terms, both are interpreted as S-finite kernels; The category of S-finite kernels is closed under composition, thus the term $\Let x=t_1\In t_2$ is also interpreted as an S-finite kernel.

The semantics of the $\Return$ statement is given by the kernel $\denote{\Return(t)}\colon \denote{\Gamma} \times \Sigma_{\denote{\TypeA}} \to [0,1]$
\[
\denote{\Return(t)}_{\gamma, A}  \defeq \begin{cases}
  1 \quad \text{if }\denote{t}_\gamma \in A \\
  0 \quad \text{otherwise}
\end{cases}
\]
Finally the $\Sample$ statement takes in as argument a deterministic term of type $\mathrm{P}(\TypeA)$ that is it takes in as argument a probability measure on the space $\denote{\TypeA}$. Thus the semantics is given as:
\[
  \denote{\Sample (t)}_{\gamma, A} =  \denote{t}_{\gamma, A}, 
\]
where $\denote{t}_\gamma\in \calM(\denote{\TypeA})$. 

\subsubsection{Soft constraints and normalization terms}
\label{sec:score;norm}
We are studying a probabilistic language for Bayesian inference; We have terms in the language that is used to scale the prior by the likelihood of some observed data and a term that re-normalizes the scaled measure to return the posterior distribution over the return type. The constructs $\Score$ and $\Normalize$ are the constructs that, respectively, scale the prior program, and normalize the program to return the posterior probability distribution on the output type, if there exists one. 
\begin{align*}
  \inference*{\Gamma |-d t\colon \Reals}{\Gamma |-p \Score(t)\colon\Unit} \qquad \inference*{\Gamma |-p t\colon \TypeA}{\Gamma |-1 \Normalize(t)\colon(\TypeA + \Unit)}.
\end{align*}
The semantics for the $\Score$ construct are given by a $S$-finite kernel on the $\Unit$ type as follows:
\begin{align*}
  \denote{\Score(a)}_{\gamma,A}= \begin{cases}
    \abs{\denote{t}_\gamma} \quad \text{if } A=\set{()}\\
    0 \qquad\quad \text{otherwise}
  \end{cases}
\end{align*}

The main difference between this semantics and the denotational semantics of the language proposed in \citet{SWY+:16, Sta:17} is in the semantics of $\Normalize$. We interpret the semantics of $\Normalize$ terms as a probability kernel on the sum space given as $\denote{\Normalize(t)}\colon \denote{\Gamma}\times\Sigma_{\TypeA+1}\to [0,1]$ defined as
\[
\denote{\Normalize(t)}_{\gamma, A} = 
\begin{cases}
  \frac{\denote{t}_{\gamma,\set{u|(0,u)\in A}}}{\denote{t}_{\gamma,\denote{\TypeA}}}\quad\text{if } \denote{t}_{\gamma,\denote{\TypeA}}\in (0,\infty)\\
0 \qquad\quad\qquad\;\; \text{else if } (1,()) \not\in A\\
1 \qquad\quad\qquad\;\; \text{else if } (1,()) (1,()) \in A
\end{cases},
\]
where $A \in\denote{\TypeA + 1}.$ The key distinction is that we are not able to determine if the term $\denote{t}_\gamma$ is an infinite measure or a null measure.

\subsubsection{Stationary terms}
One of the main contributions of this paper is that we propose a new feature in the probabilistic language that takes as argument an initial distribution and a Markov chain transition kernel on some state space, and returns the corresponding stationary distribution that leaves the kernel invariant. We allow the users to define a transition kernel on some measurable space using a standard lambda expression. The following is the syntax and typing rules for the stationary term:
\begin{align*}
  \inference*{\Gamma |-1 t_0\colon \TypeA\quad \Gamma, x:\TypeA |-1 t_1\colon \TypeA}{\Gamma |-1 \Stationary(t_0, \lambda x. t_1):\TypeA + 1}.
\end{align*}

First note that $\denote{t_0}_\gamma$ represents the initial distribution and $\denote{t_1}_\gamma$ represents the transition kernel for the Markov chain. 
The denotational semantics of $\denote{\Stationary(t_0,\lambda x.t_1)}: \denote{\Gamma}\times \Sigma_{\denote{\TypeA+1}} \to [0,1]$ is given as:
\begin{align*}
  &\denote{\Stationary(t_0,\lambda x.t_1)}_{\gamma, A} =\\
  &=\begin{cases}
    \mu(\set{u:(0,u)\in A}) &\begin{aligned}
    &\text{if } \exists! \mu\in \calM(\denote{\TypeA}): \\
    &\int_{\denote{\TypeA}} \denote{t_1}^{(n)}(x,\cdot)\to_n \mu \\
    & \text{for }x \text{ a.e. }\denote{t_0}_\gamma.
    \end{aligned}\\
    \delta_{(1,())} (A) & \text{otherwise}
  \end{cases}
\end{align*}
The function $\denote{\Stationary(t_0,\lambda x.t_1)}$ is jointly measurable. See \cref{appendix:statmeasurability} for more details.

\section{Eliminability of soft constraints and normalization terms}
\label{sec:equivalence}
We introduced a new language construct $\Stationary$, 
which takes a Markov chain transition kernel and returns the corresponding stationary distribution, and added this
construct to the language proposed in \citet{SWY+:16}. 
In this section, we compare the original language and our extension, using the framework of \defn{expressibility} due
to \citet{Fel:91}. 
In particular, we show that the soft constraint and normalization terms, presented in \cref{sec:score;norm}, are
\defn{eliminable} from the language $\langN$.

\subsection{Expressibility of programming languages}

In this section, we summarize the framework of relative expressibility of programming languages. For a detailed account, see \citet{Fel:91}. 
We begin by giving a formal definition of a programming language and conservative restrictions:
\begin{definition}[Programming Language]
  A programming language $\calL$ consists of 
  \begin{itemize}
    \item a set of $\calL$-phrases generated from a pre-specified syntax grammar, which consists of possibly infinite number of function symbols $\bbF_1, \bbF_2, \dots$ with arities $a_1, a2, \dots $ respectively; 
    \item a set of $\calL$-programs is a non-empty subset of $\calL$-phrases;
    \item a semantics to the terms of the language. 
  \end{itemize}
  A programming language $\calL'$ is a conservative restriction of a language $\calL$ if 
  \begin{itemize}
    \item the set of constructors of the language $\calL'$ is a subset of constructors of $\calL$, i.e., there exists a set $\set{\bbF_1, \bbF_2,\dots}$ that is subset of the constructors of the language $\calL$ but not in the language $\calL'.$
    \item the set of $\calL'$-phrases is the full subset of $\calL$-phrases that do not contain any constructs in $\set{\bbF_1, \dots , \bbF_n,\dots}$;
    \item the set of $\calL'$-programs is the full subset of $\calL$-programs that do not contain any constructs in ${\bbF_1,\dots, \bbF_n,\dots}$; 
    \item the semantics of $\calL'$ is a restriction of $\calL$'s semantics. 
  \end{itemize}
\end{definition}

Informally speaking, a language construct from the language is \defn{eliminable} if we can find a \defn{computable} function that: 1) maps every program in the original language to a program in its conservative restriction; 2) is ``local'' for all the language constructs that exists in the conservative restriction. We make this notion precise in the following definition: 

\begin{definition}[Eliminability of programming constructs]  \label{eliminability}
  Let $\calL$ be a programming language and let $A = \set{\bbF_1, ... , \bbF_n, ...}$  be a subset of its constructors such that $\calL'$ is a conservative restriction of $\calL$ without the constructors in $A$. The programming constructs $\bbF_i \in A$ are \defn{eliminable} if there is a computable mapping $\compiler$ from $\calL$-phrases to $\calL'$-phrases such that:
  \begin{enumerate}
    \item $\compiler(t)$ is an $\calL'$-program for all $\calL$-programs $t$;
    \item $\compiler(\bbF(t_1,\dots, t_d)) = \bbF(\compiler(t_1),\dots, \compiler(t_d))$ for any construct $\bbF$ of $\calL'$, i.e., $\compiler$ is homomorphic in all constructs of $\calL'.$
    \item For all programs $t$, 
    $\denote{t} = \denote{\compiler(t)}$.
  \end{enumerate}
\end{definition}

Note that the first two properties of eliminability (\cref{eliminability}) are syntactic, 
while the third property is semantic.

\subsection{Eliminability of the constructs ``norm'' and ``score''}

Consider the conservative restriction of $\langN$ obtained by removing the $\Normalize$ and $\Score$ constructs. 
(We denote this restriction by $\langS$, and provide its syntax in \cref{fig:lang2}.)
The main result of this section shows that $\langN$ and $\langS$ are equally expressive:

\begin{figure}[]
  \begin{align*}  
    &\defn{Types: } \\
    &\quad\TypeA_0, \TypeA_1 ::= \Reals \mid \Unit  \mid  \mathrm{P}(\TypeA) \mid  \TypeA_0 \times \TypeA_1 \mid \sum_{i\in \Nats}\TypeA_i\\
    &\defn{Terms: } \\
    &\quad\textit{deterministic:} \\
    &\quad a_0,a_1  ::= x \mid * \mid (a_0,a_1) \mid(i,a) \mid \pi_j(a) \mid f(a)\\
    &\quad \qquad \quad \mid \Case a \Of \set{(i,x)\Rightarrow a_i}_{i\in I} \\
    &\quad\textit{\uprobabilistic{}: }\\
    &\quad t_0,t_1  ::= \Sample (a) \mid \Return (a) \mid \Let x=t_0 \In t_1 \\ 
    &\quad\qquad \quad         \mid \Case a \Of \set{(i,x)\Rightarrow t_i}_{i\in I} \\
    &\quad\qquad \quad         \mid \Stationary (t_0, \lambda x.t_1) \\
    &\defn{Program: }\\
    &\quad \text{$t$ is a program if $t$ is \uprobabilistic{}}\\
    & \qquad\qquad \text{ with no free variables}
    \end{align*}
  \caption[]{Syntax for the probabilistic language: {$\langS$} }
  \label{fig:lang2}
\end{figure}

\begin{theorem}[Eliminability of $\Normalize$ and $\Score$]
  \label{thm:equivalence}
  The constructs $\Normalize$ and $\Score$ are eliminable from $\langN$. 
\end{theorem}

To prove \cref{thm:equivalence}, we give an explicit compiler $\compiler : \langN \to \langS$ that meets the requirements
of eliminability.
This compiler is given in its entirety in \cref{fig:compilerdiscr}. 
The compiler can be viewed as a simplified version of the Trace-MH algorithm first described in \citet{WSG:11}. 
The primary difference is the proposal kernel.

\begin{figure*}[tp]
    \centering
       \begin{subfigure}[t]{\textwidth}
        \centering
        \begin{align*}
        \Tracer(t) :=& \Match t \With \\
            &\quad\left\{\Sample(t)\mapsto \Sample(t), \; \Return(t)\mapsto \Return(t),\; \Stationary(t_0, \lambda x. t_1) \mapsto \Stationary(t_0, \lambda x. t_1), \right.\\
            &\quad\;\;\Score(t)\mapsto \Score(t),\; \Normalize(t) \mapsto \Normalize(t),\\
            &\;\;\quad \Let x = t_0 \In t_1 \mapsto \begin{aligned} &\Let trace_x = \Tracer(t_0) \In \\
  &\Let trace_y =\Tracer(t_1[x \backslash \pi_{-1}(trace_x)]) \In\\
  &\Return(trace_x, trace_y)\end{aligned},\\
          &\;\left.\quad\Case a \Of \set{(i,x) \implies u_i} \mapsto \Case a \Of \set{(i,x) \implies \Tracer (u_i)}\right\}
        \end{align*}
        \caption{Tracer Transformation}
        \label{fig:tracertrans}
        \end{subfigure}
    \begin{subfigure}[t]{\textwidth}
        \centering
        \begin{align*}
        \Prior\circ\Tracer(t) :=& \Match t \With\\
            &\quad\left\{\Sample(t)\mapsto \Sample(t),\;\Return(t)\mapsto \Return(t),\; \Stationary(t_0, \lambda x. t_1) \mapsto \Stationary(t_0, \lambda x. t_1),\right.\\
            &\;\;\quad\Score(t)\mapsto \Return(()),\; \Normalize(t) \mapsto \Normalize(t),\\
            &\;\;\quad \Let x = t_0 \In t_1 \mapsto\begin{aligned} &\Let trace_x = \Prior\circ\Tracer(t_0) \In \\
            &\Let trace_y =\Prior\circ\Tracer(t_1[x \backslash \pi_{-1}(trace_x)]) \In\\
            &\Return(trace_x, trace_y)
            \end{aligned},\\
          &\;\quad\Case a \Of \set{(i,x) \implies u_i} \mapsto \left.\Case a \Of \set{(i,x) \implies \Prior\circ\Tracer (u_i)}\right\}
        \end{align*}
        \caption{Prior Transformation}
        \label{fig:priortrans}
        \end{subfigure}
       \begin{subfigure}[t]{\textwidth}
        \centering
        \begin{align*}
        \Likelihood\circ\Tracer(t) := & \Match t \With \\
        &\quad\left\{\Sample(t)\mapsto 1,  \Return(t)\mapsto 1,\; \Stationary(t_0, \lambda x. t_1) \mapsto 1, \right.\\
        &\;\;\quad \Score(t)\mapsto |t|,\Normalize(t)\mapsto 1\\
        &\;\;\quad \Let x = t_0 \In t_1 \mapsto \Likelihood\circ\Tracer(t_0) * \Likelihood\circ\Tracer(t_1),\\
        &\;\left.\quad \Case a \Of \set{(i,x) \implies t_i}_{i\in I} \mapsto \Case a \Of \set{(i,x) \implies \Likelihood(t_i)}_{i\in I}\right\}
        \end{align*}
        \caption{Likelihood Transformation}
        \label{fig:lhdtrans}
        \end{subfigure}
        \begin{subfigure}[t]{\textwidth}
        \centering
        \begin{align*}
        &\MH(\Normalize(t)) := \Stationary\left(\Prior\circ\Tracer(t),
        \begin{aligned}
            &\lambda x. \Let x' = \Prior\circ\Tracer(t) \In  \\
            &\quad\;\Case \Sample(\BernoulliDist(\alpha(x, x'))) \Of \set{(0,T) \Rightarrow \Return(x') , (1,F) \Rightarrow \Return(x))}\\
        \end{aligned} \right),\\ 
        &\qquad \qquad\qquad\quad\text{ where }\alpha(x,x') = \begin{cases}
    \min \set{1, \frac{\Likelihood\circ\Tracer(t)(x')}{\Likelihood\circ\Tracer(t)(x)}} \quad \text{if }  \Likelihood\circ\Tracer(t)(x)>0 \text{ and }\Likelihood\circ\Tracer(t)(x')>0\\
    0 \qquad\qquad\quad\qquad\qquad\quad \text{otherwise}
  \end{cases}
,        \end{align*}

        \caption{Metropolis--Hastings Transformation}
        \label{fig:mhtrans}
        \end{subfigure}        
        \begin{subfigure}[t]{\textwidth}
        \centering
        \begin{align*}
        \compiler(t) := & \Match t \With \\
        &\quad\left\{\Sample(t)\mapsto \Sample(t),\;  \Return(t)\mapsto \Return(t),\; \Stationary(t_0, \lambda x. t_1) \mapsto \Stationary(t_0, \lambda x. t_1),\right.\\
        &\quad\;\; \Normalize(t) \mapsto  \pi_{-1}\circ\MH(\Normalize(t)), \Let x = t_1 \In t_2 \mapsto \Let x = \compiler(t_1) \In \compiler(t_2),\\
        &\;\left.\quad \Case a \Of \set{(i,x) \implies t_i}_{i\in I} \mapsto \Case a \Of \set{(i,x) \implies \compiler(t_i)}_{i\in I}\right\}
,        \end{align*}
        \caption{Compiler Transformation}
        \label{fig:compilertrans}
        \end{subfigure}
    \caption[]{Here we give a formal description of the compiler $\protect\Phi$ that maps programs in language $\protect\calL_{\Normalize}$ to the programs in language $\protect\langS$. In \cref{fig:tracertrans} we give the $\protect\Tracer$ transformation that modifies the sequencing term to give a joint distribution rather than marginalizing out the intermediate terms.
    The projection function, $\pi_{-1}$, returns the last element of a tuple.
    In \cref{fig:priortrans} we give a transformation that given a $\protect\langN$-term returns the prior term associated with it. In \cref{fig:lhdtrans} we give a description of the function that for a term $\protect t$, represents Radon--Nikodym derivative of $\protect\Tracer(t)$ with respect to the $\protect\Prior\circ\Tracer(t)$. In \cref{fig:mhtrans} we give \defn{independence Trace-MH} algorithm. In \cref{fig:compilertrans} we finally give the complete description of the compiler.}
    \label{fig:compilerdiscr}
\end{figure*}

First, we show that this compiler maps every program in $\langN$ to a program in $\langS$ and is homomorphic in all the features of $\langS$.
\begin{lemma}
\label{lem:compilerprop}
    Let $\compiler$ be the compiler described in \cref{fig:compilerdiscr}.  
    \begin{enumerate}
        \item $\compiler$ maps every $\langN$-program to a $\langS$ program.
        \item $\compiler$ is homomorphic in all constructs of $\langS$. 
    \end{enumerate}  
\end{lemma}

We then show $\compiler$ preserves semantics:
\begin{lemma}
\label{lem:preservessemantics}
    Let $\compiler$ be the compiler described in \cref{fig:compilerdiscr} and let $t$ be a $\langN$-program. Then, under the semantics given in \cref{sec:semantics}: 
    \[
    \norm{\denote{t} - \denote{\compiler(t)}}_\tv = 0,
    \]
    i.e., $\compiler$ preserves semantics.
\end{lemma}

\cref{thm:equivalence} follows immediately from these two lemma.
The proof of \cref{lem:compilerprop} can be found in \cref{appendix:compilerprop}.
The proof of \cref{lem:preservessemantics} occupies the remainder of the section.

\subsubsection{Compiler preserves semantics}
We begin by introducing and analyzing two source-to-source transformations, $\Tracer$ and $\Prior$, and one program transformation, $\Likelihood$, that produces a measurable function. (See \cref{fig:compilerdiscr} for their definition.)

In \cref{fig:tracertrans}, we define the \defn{tracer} transformation, $\Tracer$, 
which produces programs that return all probabilistic choices encountered in evaluation of the program as a tuple.
The behavior of this transformation is straightforward in all cases other than sequences.
On sequences, the $\Tracer$ transformation satisfies the following typing rule:
\begin{equation*}
      \inference*{\Gamma |-p \Tracer(t_0)\colon \TypeA_0\quad \Gamma, x:\TypeA |-p \Tracer(t_1)\colon \TypeA_1}{\Gamma |-p \Tracer(\Let x = t_0 \In t_1):\TypeA_0\times \TypeA_1}.
\end{equation*}
\begin{remark}
  Let $\mu\colon X\rightsquigarrow Y$ and $\nu\colon X\times Y \rightsquigarrow Z$ be $\sigma$-finite kernels.
  For $x\in X$, $A \in \Sigma_Y$, and $B\in \Sigma_Z$, define
  \[
  (\mu \otimes \nu)(x, A\times B) = \int_{A} \int_B\mu(x,dy)\nu(x,y,dz).
  \]
  By \citep[][Lem.~1.17]{Kal:17}, $(\mu \otimes \nu)$ is a $\sigma$-finite kernel. 
  Therefore, while the semantics of a probabilistic term $\Gamma|-p t\colon\TypeA$ is, in general, a $S$-finite kernel, 
  the semantics of $\denote{\Tracer(t)}$ can be taken to be a $\sigma$-finite kernel. 
\end{remark}

In \cref{fig:priortrans}, we define the \defn{prior} transformation, $\Prior$,
which removes all occurrences of $\Score$ terms that are not enclosed within a $\Normalize$ term. 
The next result establishes the key property we require of the resulting semantics:

\begin{lemma}
\label{lem:priorabc}
  Let $\Gamma |-p \Tracer(t):\TypeA$ be a \probabilistic{} term in the language $\langN$. If $\;\Prior\circ\Tracer$ is the program transformation give in \cref{fig:priortrans}, then $\Gamma |-1 \Prior\circ\Tracer(t):\TypeA$ is a \uprobabilistic{} term such that, for all $\gamma,\; \denote{\Tracer(t)}_\gamma \ll \denote{\Prior\circ\Tracer(t)}_\gamma,$ i.e., 
  \begin{equation*}
  \denote{\Prior\circ\Tracer(t)}_{\gamma, A} = 0 \implies \denote{\Tracer(t)}_{\gamma,A} = 0.
  \end{equation*} 
\end{lemma}
The proof can be found in \cref{appendix:priorabc}.

Finally, in \cref{fig:lhdtrans},
we define the likelihood transformation,  $\Likelihood$,
which produces a (deterministic) measurable function and 
satisfies the following typing rule:
\begin{equation*}
      \inference*{\Gamma |-p \Tracer(t)\colon \TypeA}{\Gamma |-d \Likelihood\circ\Tracer(t)\colon \TypeA \to \NNReals}.
\end{equation*}
The following lemma
establishes that the measurable function obtained by the likelihood transformation 
is the Radon--Nikodym derivative of (i) the semantics of tracer transformation with respect to (ii) the semantics of the prior transformation. 
\begin{lemma}
Let $\Gamma |-p t:\TypeA$ 
and $\gamma\in\denote{\Gamma}$. 
Then $\denote{\Likelihood\circ \Tracer(t)}_{\gamma}$ is the Radon--Nikodym derivative of $\denote{\Tracer(t)}_{\gamma}$ with respect to $\denote{\Prior\circ\Tracer(t)}_{\gamma}$.
\end{lemma}

\begin{proof}
  We proceed via induction on the grammar of \probabilistic{} terms. 
  The only interesting base case is $\Score(t)$.
  First, notice that
  \[ \denote{\Likelihood\circ\Tracer(\Score(t))}_\gamma = \abs{\denote{t}_\gamma}.\]
  Then, for $A\in \Sigma_{\denote{\TypeA}}$,
  \begin{align*}
      &\denote{\Tracer(\Score(t))}_{\gamma, A} = \denote{\Score(t)}_{\gamma, A}\\
      & \qquad = \int_A \abs{\denote{t}_\gamma}\denote{\Return(())}_{\gamma,dx}\\ 
      &\qquad = \int_A \abs{\denote{t}_\gamma}\denote{\Prior\circ\Tracer(\Score(t))}_{\gamma,dx}
  \end{align*}
  Thus, $\denote{\Likelihood\circ\Tracer(\Score(t))}_{\gamma}$ is the Radon--Nikodym derivative with respect to $\Prior\circ\Tracer(\Score(t))_{\gamma}$. 
  
  For the inductive step, we first consider the $\Let$ construct.
  Let $A \times B \in \Sigma_{\denote{\TypeA}} = \Sigma_{\denote{\TypeA_0}\times \denote{\TypeA_1}8}$ be a rectangle. Then
\begin{align*}
  & \denote{\Tracer(\Let x = t_0 \In t_1)}_{\gamma, A\times B} \\ 
  & = \int_{\denote{\TypeA_0}}\int_{\denote{\TypeA_1}} \denote{\Return(x', y')}_{\gamma,x', y', A\times B}\\
  &\;\;\;\quad \qquad \qquad\denote{\Tracer(t_1)}_{\gamma, x', dy'}\denote{\Tracer(t_0)}_{\gamma, dx'} \\
  &= \int_{A} \int_{B}\denote{\Tracer(t_1)}_{\gamma, x', dy'}\denote{\Tracer(t_0)}_{\gamma, dx'} \\ 
  &= \int_{A} \int_{B}\denote{\Likelihood\circ\Tracer(t_0)(x')}_{\gamma, x'}\denote{\Prior\circ\Tracer(t_0)}_{\gamma, dx'}\\
  &\;\qquad\denote{\Likelihood\circ\Tracer(t_1)(y')}_{\gamma, x', y'}\denote{\Prior\circ\Tracer(t_1)}_{\gamma, x', dy'} \\
    &= \int_{A\times B}\denote{\Likelihood\circ\Tracer(\Let x= t_0\In t_1)(x',y')}_{\gamma, x',y'} \\
  &\;\quad\qquad\denote{\Prior\circ\Tracer(\Let x = t_0 \In t_1)}_{\gamma, dx'\times dy'}.
\end{align*}
The third equality follows from the inductive hypothesis. 

It remains to consider $\Case$ statements:
\begin{align*}
  & \denote{\Tracer(\Case a \Of \set{t_i}_{i\in I})}_{\gamma, A} \\ 
  & = \denote{\Case a \Of \set{(i,x) \implies \Tracer(t_i)}_{i\in I})}_{\gamma, A}\\
  & = \denote{\Tracer(t_i)}_{\gamma,v, A} \quad \text{ if } \denote{a}_\gamma = (i,v)\\
  & = \begin{aligned}
  &\int_{A} \denote{\Likelihood\circ\Tracer(t_i)}_{\gamma, v, x}\\
  &\qquad \denote{\Prior\circ\Tracer(t_i)}_{\gamma,v, dx} 
  \end{aligned}\quad \text{ if } \denote{a}_\gamma = (i,v)\\
  & = \begin{aligned}
  &\int_{A} \denote{\Likelihood\circ\Tracer(\Case a \Of \set{(i,x) \implies t_i}_{i\in I})}_{\gamma, x}\\
  &\qquad \denote{\Prior\circ\Tracer(\Case a \Of \set{(i,x) \implies t_i}_{i\in I})}_{\gamma,dx}.
  \end{aligned}
\end{align*}
The third equality follows from the inductive hypothesis. 
\end{proof}

Given the properties of the program transformations, we show in the following proposition that the $\Normalize$ terms can be compiled into a $\Stationary$ term by constructing the Metropolis--Hastings kernel with the correct stationary distribution.

\begin{proposition}
\label{prop:MHmeas-pres}
 Let $\Gamma |-1 \Normalize(t)$ 
 and $\gamma\in \denote{\Gamma}$. Then
 \begin{equation*}
    \norm{\denote{\Normalize(t)}_{\gamma} -\denote{\pi_{-1}\circ\MH(\Normalize(t))}_\gamma}_\tv = 0,
 \end{equation*}
 i.e., the $\MH$-transformation preserves semantics.
\end{proposition}

\begin{proof}
First, note that it is straightforward to verify that
$$
    \denote{\Normalize(t)}_\gamma = \denote{\pi_{-1}\circ\Normalize(\Tracer(t))}_\gamma.
$$

Further, $\denote{\Normalize(\Tracer(t))}$ is a probability kernel that maps a context to a measure on the sum type $\denote{\TypeA + 1}$. 
We can rewrite 
\begin{align*}
\denote{\Normalize(\Tracer(t))} =&  \restrict{\denote{\Normalize(\Tracer(t))}}{\set{(0,u):u\in \denote{\TypeA}}} + \\ &\qquad \restrict{\denote{\Normalize(\Tracer(t))}}{\set{(1,())}}.
\end{align*}

We now show that
\begin{align*}
\label{eq:normtracer}
&\denote{\Normalize(\Tracer(t))} = \\
&[[\Stationary\left(\Prior\circ\Tracer(t),
        \begin{aligned}
            &\lambda x. \Let x' = \Prior\circ\Tracer(t) \In  \\
            &\;\Case \Sample(\BernoulliDist(\alpha(x, x'))) \Of \\
            &\qquad\quad\;\; (0,T) \Rightarrow \Return(x') \\
            &\quad\quad\;\;\; \mid (1,F) \Rightarrow \Return(x))
        \end{aligned} \right)]]
\end{align*}

\emph{Case 1:} $\denote{\Tracer(t)}_{\gamma, \denote{\TypeA}} \in \set{0,\infty}$:\\
  This is the case when $\denote{\Tracer(t)}$ is either a null measure or an infinite measure. Recall that $\Normalize(\Tracer(t))$ satisfies
  \begin{align*}
    \denote{\Normalize(\Tracer(t))}_{\gamma}(A) & = \restrict{\denote{\Normalize(\Tracer(t))}_\gamma}{\set{(1,())}}(A) \\
    & = \begin{cases}
    1 \quad\text{if } (1,())\in A\\
    0 \quad\text{otherwise}
     \end{cases}
  \end{align*}
  When $\denote{\Normalize(\Tracer(t))}_\gamma$ is a null measure, $\alpha(x,x') = 0$.
  Thus, there is no \emph{unique} probability measure that leaves the distribution invariant, thus $\Stationary$ returns a point measure on $(1,()).$
  When $\denote{\Normalize(\Tracer(t))}_\gamma$ is an infinite measure, the MH kernel is reversible with respect to the specified $\sigma$-finite measure. Since there is no \emph{probability} measure that leaves this distribution invariant, $\Stationary$ returns a point measure on $(1,()).$

\emph{Case 2:} $\denote{t}_{\gamma, \denote{\TypeA}} \in (0,\infty)$\\
Semantically, the Markov transition kernel in the stationary term satisfies
\begin{align*}
  \left\llbracket
  \begin{aligned}
  &\Let x' = \Prior(\Tracer(t)) \In  \\
    &\Case \Sample(\BernoulliDist(\alpha(x, x'))) \Of \\
  &\qquad (0,T) \Rightarrow \Return(x') \\
  &\quad\; \mid (1,F) \Rightarrow \Return(x))
  \end{aligned}\right\rrbracket_{\gamma, v, dv'}  \\
  = \begin{aligned}
    &\alpha(v, v')\denote{\Prior(t)}_{\gamma,dv'} + \\  
    &\quad\delta_v(dv')\int_{\denote{\TypeA}}(1-\alpha(v,y))\denote{\Prior(t)}_{\gamma}(dy),
  \end{aligned}
\end{align*}
where for $$R = \set{(x,x'): \forall y\in \set{x,x'}\Likelihood\circ\Tracer(t)(y)\in (0,\infty)}$$
$$\alpha(x, x') \mapsto \begin{cases}
    \min \set{1, \frac{\Likelihood\circ\Tracer(t)(x')}{\Likelihood\circ\Tracer(t)(x)}} \quad \text{if }  (x,x')\in R\\
    0 \qquad\qquad\quad\qquad\qquad\quad \text{otherwise}
  \end{cases}$$
  Note that the Markov kernel constructed by the program transformation is exactly the Markov kernel associated with the Metropolis--Hastings algorithm discussed in \cref{sec:MCT}. Here the proposal kernel is 
  $\denote{\Prior(\Tracer(t))}_\gamma$.
  The acceptance probability $\alpha(x,x')$ satisfies the hypotheses of 
  \cref{thm:tie98}, and thus Metropolis--Hastings kernel is reversible with respect to $\frac{\denote{\Tracer(t)}_\gamma}{\denote{\Tracer(t)}_{\gamma, \denote{\TypeA}}}$.
\end{proof}

\begin{proof}[Proof of \cref{lem:preservessemantics}]
We proceed by induction on \uprobabilistic{} terms, 
since a program is a \uprobabilistic{} without free variables. 
The only interesting base case in the $\Normalize(t)$ term. 
\cref{prop:MHmeas-pres} shows that $\compiler$ preserves semantics. 

\emph{Case 1:} $\Let$ constructor:
\begin{align*}
     \denote{\compiler(\Let x = t_0 \In t_1)}_\gamma & = \denote{\Let x = \compiler(t_0) \In \compiler(t_1)}_\gamma \\
    & = \int \denote{\compiler(t_1)}_{\gamma, x}  \denote{\compiler(t_0)}_{\gamma,dx} \\
    & = \int \denote{t_1}_{\gamma, x}  \denote{t_0}_{\gamma,dx} \\
    & = \int \denote{\Let x = t_0 \In t_1)}_{\gamma,dx},
\end{align*}
where third inequality is due to the inductive hypothesis.

\emph{Case 2}: $\Case$ constructor:
\begin{align*}
    & \denote{\compiler(\Case a \Of \set{(i,x) \implies t_i}}_{\gamma,A}  =\\
    &\qquad\qquad   = \denote{\Case a \Of \set{(i,x) \implies \compiler(t_i)}}_{\gamma,A} \\
    &\qquad\qquad   = \denote{\compiler(t_i)}_{\gamma, v,A} \quad \text{ if }\denote{a}_\gamma = (i,v) \\
    &\qquad\qquad   = \denote{t_i}_{\gamma, v,A} \quad \text{ if }\denote{a}_\gamma = (i,v) \\
    &\qquad \qquad  = \denote{\Case a \Of \set{(i,x) \implies t_i}}_{\gamma, A} ,
\end{align*}
where second inequality is due to the inductive hypothesis.
\end{proof}

\section{Approximate compilation of probabilistic programs} 
\label{sec:stationary}

\citet{SWY+:16} proposed a probabilistic programming language with constructs for $\Sample$, $\Score$ and  $\Normalize$ where the semantics of the language assumes an ideal implementations for the $\Normalize$ construct. As we saw in \cref{sec:semantics}, to compute the \uprobabilistic{} term $\Gamma |-1 \Normalize(t):\TypeA$, we need to compute the normalization factor $\denote{t}_{\gamma, \denote{\TypeA}}$; This is computationally intractable. In \cref{sec:semantics}, we introduced a new programming construct $\Stationary$ that takes a description of an initial distribution and a transition kernel, and represents the limiting distribution, if there exists one, of the associated Markov chain.  \cref{thm:equivalence} tells us that for a language with $\Stationary$, the constructs $\Normalize$ and $\Score$ are eliminable. Even though computing the stationary distribution of a Markov chain is also computationally intractable, if the Markov chain specified is \defn{ergodic}, we can iterate the Markov kernel starting from the specified initial distribution to approximate the stationary distribution. The error associated with this approximation depends on rate at which the Markov chain is converging.

In this section we concretely define this approximate compilation scheme based on iteration and state the assumption under which the Markov chain converges. A similar iteration scheme was also given in \citet{AYC:18} to approximate the fixpoint of a Markov kernel.  
In \cref{thm:quantbound} we show that if the Markov chains represented by the $\Stationary$ terms of a program in the language $\langS$ were uniformly ergodic, then we can give a quantitative bound on the error associated with this approximate compilation scheme---even when we allow for ``nested'' $\Stationary$ terms.  %

\subsection[]{Approximate implementation for $\Stationary$ terms}

In this section we give a simple approximate compilation scheme for $\Stationary$ terms similar to the approximation scheme proposed by \citet{AYC:18}. 
First, we inductively define syntactic sugar $\Iterate$ as follows:
\begin{align*}
  \Iterate^0(t_0,\lambda x. t_1) &:= t_0\\
  \Iterate^{N}(t_0, \lambda x. t_1) &:= \Let x = \Iterate^{N-1}(t_0,\lambda x. t_1) \In t_1
\end{align*}
The semantics of the approximate program transformation is given as:
\begin{align*}
  \denote{\Iterate^N(t_0, \lambda x. t_1)}_\gamma = \int \denote{t_1}^N_{\gamma,x} \denote{t_0}_\gamma(dx)
\end{align*}

Given these syntactic sugar, now we give the program transformation $\phi$ as follows:
\[
\phi(\Stationary(t_0, \lambda x. t_1), n) \defeq  \Iterate^n(t_0, \lambda x. t_1)
\]

We make the following assumption on the semantics of the Markov chain specified by the $\Stationary$ terms:
\begin{assumption}
  \label{assump:erg}
  For a term $\Gamma |-1 \Stationary(t_0, \lambda x. t_1):\TypeA+1$ in the language, for all $\gamma$, the transition kernel, $\denote{t_1}_\gamma$, and the initial distribution, $\denote{t_0}_\gamma$, specifies an ergodic Markov chain with stationary distribution $\pi_\gamma\in \calM{\denote{\TypeA}}$ such that:
  \[
  \lim_{N\to \infty}\norm{\int_{\denote{\TypeA}}\denote{t_1}_{\gamma, x}^N(\cdot) \denote{t_0}_{\gamma, dx} - \pi_\gamma(\cdot)}_\tv = 0
  \]
\end{assumption}

\begin{remark}
  Note that if \cref{assump:erg} does not hold, i.e., for the  $\Stationary$ term $\Gamma |-1 \Stationary(t_0, \lambda x. t_1'):\TypeA+1$, there is some $\gamma$, such that $\denote{t_1}_\gamma$ is not ergodic, then $\denote{\Stationary(t_0, \lambda x. t_1')}_\gamma$ is a point measure on the $\set{(1,())}$. But the program transformation $\denote{\phi(\Stationary(t_0, \lambda x. t_1)}_\gamma$ is still a probability measure on $\denote{\TypeA}$.
\end{remark}

Under \cref{assump:erg}, for a program with a single $\Stationary$ term, we know by \cref{eq:ergodic} that the iteration scheme asymptotically converges to the invariant distribution. But if we allow for possibly nested $\Stationary$ terms, there is no guarantee that the approximate compilation scheme converges to something meaningful. This problem is highlighted below. 
\begin{problem}
  \label{problem:1}
  One of the main challenges when studying approximate implementation of the $\Stationary$ construct is that it is not continuous, i.e., for some term $\Gamma |-1 \Stationary(t_0, \lambda x.t_1):\TypeA+1$ if we know that $\denote{t_1}_{\gamma, x}$ is an ergodic kernel that has a unique stationary distribution, it is possible to construct an approximate implementation of the Markov transition kernel $\lambda x. t_1'$ such that 
  \[
  \exists \delta\in(0,1)\forall \gamma,x. \norm{\denote{t_1}_{\gamma,x} - \denote{t_1'}_{\gamma,x}} \leq \delta,
  \] but the Markov kernel $\denote{t_1'}_{\gamma,x}$ does not have a stationary distribution. Such an example is given in Proposition 1 of \citet{RRS:98}.
\end{problem}

To side step the issue stated in \cref{problem:1}, we need to make further semantic restrictions on the Markov chains specified by the $\Stationary$ terms. We identify that if the transition kernel given to the $\Stationary$ construct is \emph{uniformly ergodic}, then $\Stationary$ construct is continuous. 

We define this semantic restriction as relation that is assumed to be true when we give the quantitative error bounds to probabilistic program.
\begin{definition}
  Let $R \subseteq \calM(\TypeA)^{\denote{\TypeA}} \times \calM{\denote{\TypeA}} \times \NNReals \times [0,1)$ be a relation such that $R(P, \pi, C, \rho)$ holds if and only if $P$ is a uniformly ergodic Markov chain with stationary distribution $\pi$ such that for all $x\in \denote{\TypeA}$
  \[
    \norm{P^N(x, \cdot) - \pi(\cdot)}_{\tv} \leq C \rho^N
  \]
\end{definition}

\subsection{Quantitative error bounds}
Here, we derive the quantitative error bounds associated with the approximation compilation scheme of probabilistic programs where each $\Stationary$ term specifies a Markov chain that is uniformly ergodic. The main theorem we prove in this section is the following:
\begin{theorem}[Quantitative error bound for probabilistic programs]
  \label{thm:quantbound}
  Let $P$ be a probabilistic program in the language $\langS$. Let $\set{\Stationary(t_{0i}, \lambda x.t_{1i})}_{i\in I}$ be the set of all stationary terms in the program $\varnothing |-1 P:\TypeB$ such that for all $\gamma$, there exist constants $\set{C_i}$ and $\set{\rho_i}$ such that $R(\denote{t_{1i}}_{\gamma}, \denote{\Stationary(t_{0i},\lambda x. t_{1i})}_\gamma, C_i, \rho_i)$ holds. 
  Let $P'$ be a program where for all $i\in I$ and $N_i\in \Nats$, $\Stationary (t_{0i}, \lambda x.t_{1i})$ is replaced by $\phi(\Stationary (t_{0i}, \lambda x.t_{1i}), N_i)$. Then, there exist constants $\set{C_i'}_{i\in I}$ such that
  \[
    \norm{\denote{P}_\gamma - \denote{P'}_\gamma}_\tv \leq \sum_{i\in I}C_i'\rho_i^{N_i}.
  \]
\end{theorem}
In the remainder of this section we prove the result above. We begin by first establishing the uniform continuity of $\Let$ and $\Case $ constructs. 
\begin{proposition}
  The following statements hold:
  \begin{enumerate}
    \item Let $\Gamma |-1 t_0:\TypeA$, $\Gamma |-1 t_0':\TypeA$, $\Gamma, x:\TypeA |-1 t_1:\TypeB$, and $\Gamma , x:\TypeA |-1 t_1':\TypeB$ be \uprobabilistic{} terms. If for all $\gamma\in \Gamma$, $\norm{\denote{t_0'}_\gamma -\denote{t_0}_\gamma}_{\tv} \leq \alpha$ and for all $\denote{x}_\gamma\in \denote{\TypeA}$,  $$\norm{\denote{t_1'}_\gamma(\denote{x}_{\gamma}) -\denote{t_1}_\gamma(\denote{x}_\gamma)}_{\tv} \leq \beta$$ then $$\norm{\denote{\Let x = t_0 \In t_1}_\gamma - \denote{\Let x = t_0' \In t_1'}_\gamma}_{\tv}\leq \alpha + \beta$$
    
    \item Let $\Gamma, x:\TypeA_i |-1 t_i:\TypeB$ and $\Gamma, x:\TypeA_i |-1 t'_i:\TypeB$, such that $$\forall i\in I, \forall x\in \denote{\TypeA_i},  \norm{\denote{t'_{i}}_{\gamma,x} -\denote{t_i}_{\gamma,x}}_\tv \leq \alpha_i$$      
  \end{enumerate}

  \begin{equation*}
    \left\lVert
    \begin{aligned}
      &\denote{\Case a \In \set{(i,x) \Rightarrow t'_i}_{i\in I}}_\gamma - \\
      &\qquad \denote{\Case a \In \set{(i,x) \Rightarrow t_i}_{i\in I}}_\gamma  
    \end{aligned}
      \right\rVert_{\mathrm{tv}}
      \leq \sup \set{\alpha_i}_{i\in I}
  \end{equation*}

  \label{prop:let-case}
\end{proposition}
\begin{proof}
Proof in \cref{appendix:let-case}
\end{proof}

The next theorem quantifies the error associated with the $N$ step $\Iterate$ transformation of the $\Stationary$ term with an \emph{approximate} transition kernel to the $\Stationary$ term with the \emph{correct} transition kernel. 
\begin{theorem}
  \label{thm:approxstat}
  Let $\Gamma , x |-1 t_1$ and $\Gamma , x |-1 t_1'$ be probabilistic terms. If there exist $\pi$, $C$, and  $\rho$ such that $R(\denote{t_1}_\gamma, \pi, C, \rho)$ holds and \[\norm{\denote{t_1'}_\gamma - \denote{t_1}_\gamma}_\tv \leq \veps,\] then
  \[
  \norm{\denote{\Iterate^N(t_0, \lambda x. t_1')} - \denote{\Stationary(t_0, \lambda x. t_1)}}_\tv \leq \frac{\veps C}{1-\rho} + C\rho^N.\]
\end{theorem}
To prove the theorem above, we begin by first giving a simple contraction lemma that quantifies the total variation distance between the laws of the $n^{\text{th}}$ random variable of Markov chains with different initial distribution but same transition kernel under the assumption of uniform ergodicity. 
\begin{lemma}[Contraction]
  \label{lem:2}
  Let $\Gamma |-1 t_0:\TypeA$ and $\Gamma, x:\TypeA |-1 t_1:\TypeA$ be \uprobabilistic{} terms. If these terms are such that  $R\left(\denote{t_1}_\gamma, (\denote{\Stationary(t_0,\lambda x.t_1)}_{\gamma}), C, \rho\right)$ holds, then for all 
  \begin{equation*}
    \left\lVert 
    \begin{aligned}
      &\denote{\Iterate^N(m_1, \lambda x. t_1)}_\gamma - \\
      &\qquad\denote{\Iterate^N(m_2, \lambda x. t_1)}_\gamma
    \end{aligned}\right\rVert_{\mathrm{tv}} \leq C\rho^N\norm{\denote{m_1}_\gamma - \denote{m_2}_\gamma}.    
  \end{equation*}
\end{lemma}

\begin{proof}
  The proof of this lemma follows directly from linearity of integration.
\end{proof}

To prove \cref{thm:approxstat}, we now give following theorem that quantifies the distance between the semantics of the $\Iterate$ transformation when the transition kernels are close.
\begin{lemma}
  \label{lemma:iter}
  Let $\Gamma , x |-1 t_1$ and $\Gamma , x |-1 t_1'$ be probabilistic terms. If there exist $\pi \in \calM(\TypeA)$, $C\in \NNReals$, and $\rho\in [0,1)$ such that $R(\denote{t_1}_\gamma, \pi, C, \rho)$ and $$\norm{\denote{t_1'}_\gamma - \denote{t_1}_\gamma}_\tv \leq \veps$$ then, 
\[
  \norm{\denote{\Iterate^N(t_0, \lambda x. t_1')} - \denote{\Iterate^N(t_0, \lambda x. t_1)}}_\tv \leq \frac{\veps C}{1-\rho},
\]
\end{lemma}
\begin{proof}
  We show this by first noting that 
  \begin{align*}  
    & \norm{\denote{\Iterate^N(t_0, \lambda x. t_1))}_\gamma - \denote{\Iterate^N(t_0, \lambda x. t_1')}_\gamma}_{\tv}   
  \end{align*}  

  \begin{equation*}
    = \left\lVert
    \begin{aligned}
    &\sum_{i=0}^{N-1}\denote{\Iterate^{N-i}(\Iterate^{i}(t_0, \lambda x. t_1'), \lambda x. t_1) }_\gamma \\
    &\qquad\quad - \denote{\Iterate^{N-i-1}(\Iterate^{i+1}(t_0, \lambda x. t_1'), \lambda x. t_1) }_\gamma
    \end{aligned} 
    \right\rVert_{\tv}
  \end{equation*}
  \begin{equation*}
    = \left\lVert
    \begin{aligned}
    &\sum_{i=0}^{N-1} \denote{\Iterate^{N-i-1}(\Let x = \Iterate^{i}(t_0, \lambda x. t_1') \In n, \lambda x. t_1))} \\
    &\quad \qquad- \denote{\Iterate^{N-i-1}(\Iterate^{i+1}(t_0, \lambda x. t_1'), \lambda x. t_1)}_\gamma
    \end{aligned} 
    \right\rVert_{\tv}
  \end{equation*}

  Applying the contraction lemma,
  \begin{align*}
    \leq & \sum_{i=0}^{N-1} C\rho^{N-i-1}\left(\left\lVert\denote{\Let x = \Iterate^{i}(t_0, \lambda x. t_1') \In n}_\gamma \right.\right. \\
    & \left.\left.\qquad\qquad\qquad\quad - \denote{\Let x = \Iterate^{i}(t_0, \lambda x. t_1') \In n'}_\gamma\right\rVert_\tv \right)\\
    \leq& \sum_{i=0}^{N-1} C\rho^{i} \veps \leq \frac{\veps C}{1-\rho}.
  \end{align*}
\end{proof}

\begin{proof}[Proof of \cref{thm:approxstat}]
Given the \cref{lemma:iter} and the assumption that $R(\denote{t_1}_\gamma, \pi, C, \rho)$ holds, the result follows by the triangle inequality.
\end{proof}

Now we prove the main theorem of this section.
\begin{proof}[Proof of \cref{thm:quantbound}]
  We proceed by induction on probabilistic terms.
    \begin{itemize}
    \item \emph{Base case:}  \\ 
      Leaf node for the induction is the terms of the form $\Stationary(t_0', \lambda x. t_1')$ .
      By the assumption of uniform ergodicity there exist $C', \rho'$ such that the following holds $$\denote{\Stationary(t_0, \lambda x. t_1)} - \denote{\Iterate^{N'}(t_0, \lambda x. t_1)}\leq C'\rho'^{N'}.$$ 
      Thus the base case holds. 
    \item \emph{Inductive Step:} \\
      For the inductive step we show the hypothesis holds for all constructor of probabilistic terms in our language. 
      \begin{itemize}
        \item $\Case$ terms: By the inductive hypothesis, $$\norm{\denote{t'_j}_\gamma - \denote{t'_j}_\gamma}_\tv \leq \sum_{i\in I_j}C_i\rho^{N_i}.$$
        Now, we show
        \begin{equation*}
          \left\lVert
          \begin{aligned}
            & \denote{\Case a \In \set{(j,x) \Rightarrow t'_j}_{j\in J}} - \\
            & \qquad \denote{\Case a \In \set{(j,x) \Rightarrow t_j}_{j\in J}}
          \end{aligned}\right\rVert_{\tv}  \leq \sum_{i\in \bigcup_{j\in J}I_j} C'_i \rho_i^{N_i}     
        \end{equation*}
      
        From \cref{prop:let-case} and the inductive hypothesis, we know

        \begin{align*}
          \left\lVert
          \begin{aligned}
            & \denote{\Case a \In \set{(j,x) \Rightarrow t'_j}_{j\in J}} - \\
            & \qquad - \denote{\Case a \In \set{(j,x) \Rightarrow t_j}_{j\in J}}
          \end{aligned}\right\rVert_{\tv}  &\leq \sup_{j\in J} \sum_{i\in I_j}C_i\rho^{N_i}\\
          &\leq \sum_{i\in \bigcup_j I_j}C_i\rho^{N_i}
        \end{align*}

        \item $\Let$term: We need to show $$\norm{\denote{\Let x=t_1 \In t_2} - \denote{\Let x=t'_1 \In t'_2}}_\tv \leq \sum_{i\in I_1\cup I_2}C_i \rho_i^{N_i}$$
        This follows from inductive hypothesis and \cref{prop:let-case}.
        \item $\Stationary$ term: From \cref{thm:approxstat} and inductive hypothesis it follows that 
        \begin{equation*}
          \left\lVert \begin{aligned}
            &\denote{\Stationary(t_0, \lambda x. t_1)} - \\
            &\qquad \denote{\Iterate^{N'}(t_0, \lambda x.t_1')}
          \end{aligned}
          \right\rVert_{\tv}\leq C'\rho'^{N_i}  + \sum_{i} \frac{C'C_i}{1-\rho'} \rho_i^{N_i}.
        \end{equation*}
      \end{itemize}
    \end{itemize}
\end{proof}

\section{Summary and Discussion}
MCMC algorithms are workhorses for approximate inference in probabilistic models. MCMC algorithms are popular because they give us asymptotic convergence guarantees and are commonly used as ``approximate'' compilers in probabilistic programming languages. In this paper we proposed a language construct $\Stationary$ that allows us give a formal description for such compilers. We then gave a simple compiler description that, at its core, implements an MCMC algorithm. 
Typically quantifying the rate at which Markov chain, with a given transition kernel, converges is an open problem and the one we do not attempt to solve in this paper. We make a semantic assumption that the user using our language provides us with a description of the Markov kernel that converges uniformly to the corresponding target distribution. 
We show that under this uniform convergence property, we can derive rates at which the approximate compiler converges to the original program. 

The assumption for uniform ergodicity is crucial for us to derive the quantitative bound. The main difficulty we found in relaxing the uniform ergodicity assumption is the fact that our language allows us to nest the $\Stationary$. We leave as open problem if we can relax the uniform ergodicity assumption. 

\bibliography{references}
\pagebreak
\clearpage
\appendix

\section[]{Measurability of $\Stationary$}
\label{appendix:statmeasurability}
The map $\denote{\Stationary(t_0,\lambda x.t_1)}\colon\denote{\Gamma}\to \calM(\denote{\TypeA})$ is measurable. We show this by constructing a test function $\phi_{n}$ as follows:
\begin{align*}
    \phi_{n}(x,x') = \norm{\denote{t_1}^{(n)}(x, \cdot) - \denote{t_1}^{(n)}(x', \cdot)}_\tv
\end{align*}
such that $\exists! \mu\in \calM(\denote{\TypeA}): \int_{\denote{\TypeA}} \denote{t_1}^{(n)}(x,\cdot)\to_n \mu $ for 
$ x$ a.e. $\denote{t_0}_\gamma$ if and only if $\int \denote{t_1}^{(n)}(x,\cdot)$ is Cauchy in the total variation metric and $\lim_{n\to \infty} \phi_n(x,x') = 0$  $(x,x')$ a.e. $\denote{t_0}_\gamma\otimes\denote{t_0}_\gamma$.

\section[]{Proof of \cref{lem:compilerprop}}
\label{appendix:compilerprop}
We need show that the syntactic manipulations $\compiler$ defined in \cref{fig:compilerdiscr} satisfy the following two properties: 
\begin{enumerate}
    \item $\compiler(t)$ is a $\langS$-program for all $\langN$-program $t$.
    \item $\compiler(\bbF(a_1, ..., a_n)) = \bbF(\compiler(a_1), ..., \compiler(a_n))$ for all constructors in $\langS$
 \end{enumerate}

Recall that a program in $\langN$ and $\langS$ is a \uprobabilistic{} term with no free variables. To verify the first property, we need to make sure that for every program $|-1 t:\TypeA$, $\compiler(t)$ is a also a \uprobabilistic{} term with no free variables. This is an easy fact to verify by inducting on \uprobabilistic{} terms. You can also witness this fact by noticing that every term in $\langS$ is \uprobabilistic{} and the compiler does not introduce any free variables. 

For the second property, first recall compiler does not modify any deterministic term and that the set of probabilistic constructs in $\langS$ are as follows: $$\calF = \set{\Sample, \Return, \Let, \Case, \Stationary}.$$ 
By examining the the compiler transformation for all of these terms in \cref{fig:compilerdiscr}, we notice that the compiler is homomorphic in all constructs of $\langS$.

\section[]{Proof of \cref{lem:priorabc}}
\label{appendix:priorabc}
  The proof that $\Prior(t)$ is a \uprobabilistic{} term can be witnessed by the fact that $\Prior$ transforms all $\Score$ statements to $\Return$ statements. 
  Since probabilistic terms are closed under composition, it is a simple exercise to show that $\Prior(t)$ is \uprobabilistic{}. 
  We now show that for all $\gamma$
  \[\denote{\Prior(t)}_{\gamma, A} = 0 \implies \denote{t}_{\gamma,A} = 0\]
  This is shown by induction on terms.
  For the base cases the only change is in the $\Score$-term where for some arbitrary $\gamma$,
  \begin{align*}
    &\denote{\Prior(\Score(t))}_{\gamma, A} = \denote{\Return(())}_{\gamma, A} = \begin{cases}
      1 \quad \text{if } ()\in A\\
      0 \quad \text{if } ()\not\in A 
    \end{cases} \\
    &\denote{\Score(t)}_{\gamma, A} = \begin{cases}
      \abs{\denote{t}_\gamma} \quad \text{if } ()\in A\\
      0 \quad \text{if } ()\not\in A.
    \end{cases}
  \end{align*}
Thus, if $\denote{\Prior(\Score(t))}_{\gamma, A} = 0$ implies $()\not\in A$ which means $\denote{\Score(t)}_{\gamma, A} = 0$. Hence we have shown $\denote{\Score(t)}_{\gamma}\ll\denote{\Prior(\Score(t)}_{\gamma}.$ This finishes the bases cases for the induction. 

For the inductive step, we need to show that $\Let$terms and $\Case$terms is absolutely continuous with respect to the prior transformation for $\Let$terms and $\Case$terms respectively.

First we prove the let statements. For $A\in \denote{\TypeA_1}$,
\begin{align*}
  &\denote{\Prior(\Let x=t_0 \In t_1)}_{\gamma,A} \\
  &= \denote{\Let x=\Prior(t_0) \In \Prior(t_1)}_{\gamma,A}\\
  &= \int_{A}\int_{\denote{\TypeA_0}} \denote{\Prior(t_1)}_{\gamma, x, dy}\denote{\Prior(t_0)}_{\gamma, dx}
\end{align*}

By IH, $\set{\denote{t_i}_{\gamma} \ll \denote{\Prior(t_i)}_{\gamma}}_{i\in \set{0,1}}$. 
\cref{rnthm} guarantees us that there exist measurable functions $f_0\colon\denote{\Gamma}\times \denote{\TypeA_0} \to [0,\infty]$, and $f_1\colon\denote{\Gamma}\times \denote{\TypeA_0} \times \denote{\TypeA_1} \to [0,\infty]$ such that 
\begin{align*}
    &\denote{t_0}_{\gamma, A} = \int_A f_0(\gamma, x)\denote{\Prior(t_0)}_{\gamma, dx}, \text{ and } \\ &\denote{t_1}_{\gamma,x, A} = \int_A f_1(\gamma,x, y)\denote{\Prior(t_1)}_{\gamma, dy}.
\end{align*}
Now, let $A\in \Sigma_{\denote{\TypeA}}$ be such that: 
\begin{align*}
  \int_{A}\int_{\denote{\TypeA_0}} \denote{\Prior(t_1)}_{\gamma, x, dy}\denote{\Prior(t_0)}_{\gamma, dx} = 0
\end{align*}
Also, we know 
\begin{align*}
  &\denote{\Let x=t_0 \In t_1}_{\gamma,A} \\
  &= \int_{A}\int_{\denote{\TypeA_0}}f_0(\gamma, x) f_1(\gamma, x, y)\denote{\Prior(t_1)}_{\gamma, x, dy}\denote{\Prior(t_0)}_{\gamma, dx}.
\end{align*}
Since the integration of any measurable function on a null set is 0; Thus the following statement holds:
\begin{align*}
\int_{A}\int_{\denote{\TypeA_0}}f_0(\gamma, x) f_1(\gamma, x, y)\denote{\Prior(t_1)}_{\gamma, x, dy}\denote{\Prior(t_0)}_{\gamma, dx} = 0.
\end{align*}

This concludes the proof for $\Let$statements.
Now we show for the $\Case$ statements:
\begin{align*}
  &\denote{\Prior(\Case a \Of \set{(i,x)\implies t_i})}_{\gamma,A} \\
  & =  \denote{\Case a \Of \set{(i,x)\implies \Prior(t_i)}}_{\gamma,A}\\
  & = \denote{\Prior(t_i)}_{\gamma,d,A} \quad \text{if } \denote{a}_\gamma = (i,d)
\end{align*}
We know, 
\begin{align*}
    &\denote{\Case a \Of \set{(i,x)\implies t_i}}_{\gamma,A} \\
    &=  \denote{t_i}_{\gamma,d,A} \quad \text{if } \denote{a}_\gamma = (i,d)
\end{align*}
But by IH, $\denote{t_i}_{\gamma,d} \ll \denote{\Prior(t_i)}_{\gamma,d}$.  Thus, 
\begin{align*}
&\denote{\Case a \Of \set{(i,x)\implies t_i}}_{\gamma} \\
&\ll \denote{\Prior(\Case a \Of \set{(i,x)\implies t_i})}_{\gamma},
\end{align*}
completing the proof.

\section[]{Proof of \cref{prop:let-case}}
\label{appendix:let-case}
  \begin{enumerate}
    \item For the $\Let$construct: 
    \begin{align*}
      &\norm{\int \denote{t_1}_{\gamma,x} \denote{t_0}_{\gamma,dx} - \int \denote{t_1'}_{\gamma,x} \denote{t_0'}_{\gamma,dx}}_{\tv}  \\
      \leq & \norm{\int \denote{t_1}_{\gamma,x} \denote{t_0}_{\gamma,dx}- \int \denote{t_1'}_{\gamma,x} \denote{t_0}_{\gamma,dx}}_{\tv} \\
      & \qquad\qquad + \norm{\int \denote{t_1'}_{\gamma,x} \denote{t_0}_{\gamma,dx} - \int \denote{t_1'}_{\gamma,x} \denote{t_0'}_{\gamma,dx}}_{\tv} \\
      = & \norm{\int \left(\denote{t_1}_{\gamma,x} -  \denote{t_1'}_{\gamma,x}\right) \denote{t_0}_{\gamma,dx}}_{\tv} \\ 
      & \quad + \sup_{A}\abs{\int \denote{t_1'}_\gamma(x, A) \denote{t_0}_{\gamma,dx} - \int \denote{t_1'}_\gamma(x, A) \denote{t_0'}_{\gamma,dx}}\\
      \leq & \int \sup_{x'}\norm{\denote{t_1}_\gamma(x') -  \denote{t_1'}_\gamma(x')}_{\tv}\denote{t_0}_{\gamma,dx} \\ 
      & \qquad\qquad + \sup_{f\leq 1}\abs{\int f(x) \denote{t_0}_{\gamma,dx} - \int f(x) \denote{t_0'}_{\gamma,dx}}\\
      =& \int \beta \denote{t_0}_{\gamma,dx} + \alpha \\
      = &\alpha + \beta
    \end{align*}
    \item For the $\Case$ construct:
    \begin{equation*}
      \left\lVert
      \begin{aligned}
        & \denote{\Case a \In \set{(i,x) \Rightarrow t'_i}_{i\in I}}_\gamma - \\
        &\quad\qquad \denote{\Case a \In \set{(i,x) \Rightarrow t_i}_{i\in I}}_\gamma
      \end{aligned}
      \right\rVert_{\mathrm{tv}}
    \end{equation*}
    \begin{align*}
      & = \norm{    
        \denote{t'_i}_{v, \gamma} - \denote{t_i}_{v,\gamma}}_\tv \qquad\text{if }(i,v)= \denote{a}_\gamma \\
        & \leq \alpha_i \qquad \qquad\qquad\qquad\quad\;\;\text{if }(i,v)= \denote{a}_\gamma \\
        & \leq \sup_i\set{\alpha_i}_{i\in I}
    \end{align*}
  \end{enumerate}
\end{document}